\documentclass{llncs}
\usepackage[pst,vaucanson-g]{xcolor}
\usepackage{vaucanson-g}
\usepackage{pstricks}
\usepackage{pst-tree,pst-node}
\usepackage{amsfonts}
\usepackage{url}
\usepackage{array}
\usepackage[labelsep=quad,indention=10pt]{subfig}
\usepackage{makecell}
\usepackage{tikz}
\usepackage{colortbl}
  \newcommand{\defproblem}[3]{
  \vspace{2mm}
\noindent\fbox{
  \begin{minipage}{0.96\textwidth}
  #1\\
  {\bf{Input:}} #2  \\
  {\bf{Output:}} #3
  \end{minipage}
  }
  \vspace{2mm}
}

\usepackage[algo2e, noend, noline, boxed]{algorithm2e}
  \SetKwComment{tcm}{\{}{\}}
   \newenvironment{myfunction}[2][htbp]
  {%
    \setlength{\algomargin}{.2cm}
    \begin{center}
    \begin{minipage}{#2}
    \begin{function}[#1]
    \small
     \let\Par=\par
       \def\par{\endgraf\vspace{.1cm}}
           \SetKw{To}{to}%
       \SetKw{Downto}{downto}%
           \SetKw{Or}{or}%
       \SetKwFor{Algo}{Function}{}{}%
      \vspace{.15cm}%
   }
   {%
     \let\par=\Par\end{function}%
     \end{minipage}%
     \end{center}%
   }

\newcommand{\MAW}{\textsc{MinimalAbsentWords}}
\def\dd{\mathinner{.\,.}}
\newcommand{\cO}{\mathcal{O}}
\newcommand{\SA}{\textsf{SA}}
\newcommand{\iSA}{\textsf{iSA}}
\newcommand{\LCP}{\textsf{LCP}}
\newcommand{\lcp}{\textsf{lcp}}
\newcommand{\Before}{\textsf{Before}}
\newcommand{\BeforeLCP}{\textsf{BeforeLCP}}
\newcommand{\Interval}{\textsf{Interval}}
\newcommand{\Lifolcp}{\textsf{LifoLCP}}
\newcommand{\Liforem}{\textsf{LifoRem}}

\begin{document}
\frontmatter          
%
%
\title{Linear-time Computation of Minimal Absent Words Using Suffix Array}

\author{Carl Barton\inst{1}
\and Alice Heliou\inst{2,3}
\and Laurent Mouchard\inst{4}
\and Solon P.\ Pissis\inst{1}
}

\institute{$\!^1\ $Department of Informatics, King's College London, London, UK \\ 
\email{\{carl.barton,solon.pissis\}@kcl.ac.uk} \\ 
$\!^2\ $Inria Saclay-\^{I}le de France, AMIB,  B\^{a}timent Alan Turing, France\\
$\!^3\ $Laboratoire d'Informatique de l'\'{E}cole Polytechnique (LIX), CNRS UMR 7161, France\\
\email{alice.heliou@polytechnique.org}\\
$\!^4\ $University of Rouen, LITIS EA 4108, TIBS, Rouen, France\\
\email{laurent.mouchard@univ-rouen.fr}}

\maketitle

\begin{abstract}
  An \textit{absent word} of a word $y$ of length $n$ is a word that does not occur in $y$.
  It is a \textit{minimal absent word} if all its proper factors occur in $y$. 
  Minimal absent words have been computed in genomes of organisms from all domains of life; their computation provides a fast alternative for measuring approximation in sequence comparison.
  There exists an $\cO(n)$-time and $\cO(n)$-space algorithm for computing all minimal absent words on a fixed-sized alphabet based on the construction of suffix automata (Crochemore et al., 1998).
  No implementation of this algorithm is publicly available. There also exists an $\cO(n^2)$-time and $\cO(n)$-space algorithm for the same problem based on the construction of suffix arrays (Pinho et al., 2009). 
  An implementation of this algorithm was also provided by the authors and is currently the fastest available.
  In this article, we bridge this unpleasant gap by presenting an $\cO(n)$-time and $\cO(n)$-space algorithm for computing all minimal absent words based on the construction of suffix arrays.
  Experimental results using real and synthetic data show that the respective implementation outperforms the one by Pinho et al.
\end{abstract}
\section{Introduction}
\label{sec:intro}
  Sequence comparison is an important step in many important tasks in bioinformatics.
  It is used in many applications; from phylogenies reconstruction to the reconstruction of genomes.
  Traditional techniques for measuring approximation in sequence comparison are based on the notions of distance or of similarity between sequences;
  and these are computed through sequence alignment techniques.
  An issue with using alignment techniques is that they are computationally expensive: they require quadratic time in the length of the sequences.
  Moreover, in molecular taxonomy and phylogeny, for instance, whole-genome alignment proves both computationally expensive and hardly significant.
  These observations have led to increased research into \textit{alignment free} techniques for sequence comparison.
  A number of alignment free techniques have been proposed: in~\cite{HauboldPMW05}, a 
  method based on the computation of the shortest unique factors of each sequence is proposed; other approaches estimate the number of mismatches per site based on the length 
  of exact matches between pairs of sequences~\cite{Domazet-Loso:2009:EEP:1671627.1671629}.
  
  Thus standard notions are gradually being complemented (or even supplanted) by other measures that refer, implicitly or explicitly, to the 
  composition of sequences in terms of their constituent patterns. One such measure is the notion of words absent in a sequence.
  A word is an \textit{absent word} of some sequence if it does not occur in the sequence. 
  These words represent a type of \textit{negative information}: information about what does not occur in the sequence.
  Noting the words which do occur in one sequence but do not occur in another can be used to detect mutations or other biologically significant events.

  Given a sequence of length $n$, the number of absent words of length at most $n$ can be exponential in $n$, meaning that using all the absent words for 
  sequence comparison is more expensive than alignments. However, the number of certain subsets of absent words is only linear in $n$.
  An absent word of a sequence is a \textit{shortest absent word} if all words shorter than it do occur in the sequence.
  An $\cO(mn)$-time algorithm for computing shortest absent words was presented in~\cite{Hampikian_absentsequences:}, 
  where $m$ is a user-specified threshold on the length of the shortest absent words. This was later improved by~\cite{abwords}, who presented an $\cO(n \log \log n)$-time algorithm for the same problem. 
  This has been further improved and an $\cO(n)$-time algorithm was presented in~\cite{Wu2010596}.
  
  A \textit{minimal absent word} of a sequence is an absent word whose proper factors all occur in the sequence.
  Notice that minimal absent words are a superset of shortest absent words~\cite{Pinho2009}. 
  An upper bound on the number of minimal absent words is $\cO(\sigma n)$~\cite{Crochemore98automataand,Mignosi02}, where $\sigma$ is the size of the alphabet.
  This suggests that it may be possible to compare sequences in time proportional to their lengths, for a fixed-sized alphabet, instead of proportional to
  the product of their lengths~\cite{HauboldPMW05}. 
  
  Recently, there has been a number of biological studies on the significance of absent words.
  The most comprehensive study on the significance of absent words is probably~\cite{nullrly}; in this, the authors suggest that the deficit of certain subsets of absent words 
  in vertebrates may be explained by the hypermutability of the genome. It was later found in~\cite{minabpro} that the compositional biases observed in~\cite{nullrly} for vertebrates 
  are not uniform throughout different sets of minimal absent words. Moreover, the analyses in~\cite{minabpro} support the hypothesis of the inheritance of minimal absent words through a common ancestor, 
  in addition to lineage-specific inheritance, only in vertebrates. In~\cite{citeulike:10180507}, the minimal absent words in four human genomes were computed, and it was shown that, as expected, 
  intra-species variations in minimal absent words were lower than inter-species variations. Minimal absent words have also been used for phylogenies reconstruction~\cite{Chairungsee2012109}.
  
  From an algorithmic perspective, an $\cO(n)$-time and $\cO(n)$-space algorithm for computing all minimal absent words on a fixed-sized alphabet based on the construction of suffix automata was presented 
  in~\cite{Crochemore98automataand}. An alternative $\cO(n)$-time solution for finding minimal absent words of length at most $\ell$, such that $\ell = \cO(1)$, based on the construction of tries of 
  bounded-length factors was presented in~\cite{Chairungsee2012109}. A drawback of these approaches, in practical terms, is that the construction of suffix automata (or of tries) 
  may have a large memory footprint. Due to this, an important problem is to be able to compute the minimal absent words of a sequence without the use of data structures such as the suffix automaton.
  To this end, the computation of minimal absent words based on the construction of suffix arrays was considered in~\cite{Pinho2009}; although fast in practice, the worst-case runtime of this algorithm is $\cO(n^2)$.
  Alternatively, one could make use of the succinct representations of the bidirectional BWT, recenlty presented in~\cite{Belazzougui2013}, to compute all minimal absent words in time $\cO(n)$.
  However, an implementation of these representations was not made available by the authors; and it is also rather unlikely that such an implementation will outperform an $\cO(n)$-time algorithm based on
  the construction of suffix arrays.

  \noindent \textit{Our Contribution}: In this article, we bridge this unpleasant gap by presenting the first $\cO(n)$-time and $\cO(n)$-space algorithm for computing all minimal absent words of a sequence of length $n$ 
  based on the construction of suffix arrays. The respective implementation is also provided and shown to be more efficient than existing tools.
\section{Definitions and Notation}
  To provide an overview of our result and algorithm, we begin with a few
definitions.
  Let $y=y[0]y[1]\dd y[n-1]$ be a \textit{word} of \textit{length} $n=|y|$
over a finite ordered \textit{alphabet} $\Sigma$ of size 
$\sigma = |\Sigma|=\cO(1)$.
  We denote by $y[i\dd j]=y[i]\dd y[j]$ the \textit{factor} of $y$ that 
starts at position $i$ and ends at position $j$ and by $\varepsilon$ 
the \textit{empty word}, word of length 0. 
  We recall that a prefix of $y$ is a factor that starts at position 0 
($y[0\dd j]$) and a suffix is a factor that ends at position $n-1$ 
($y[i\dd n-1]$), and that a factor of $y$ is a \textit{proper} factor if 
it is not the empty word or $y$ itself.

  Let $x$ be a word of length $0<m\leq n$. 
  We say that there exists an \textit{occurrence} of $x$ in $y$, or, more 
simply, that $x$ \textit{occurs in} $y$, when $x$ is a factor of $y$.
  Every occurrence of $x$ can be characterised by a starting position in $y$. 
  Thus we say that $x$ occurs at the \textit{starting position} $i$ in $y$ 
when $x=y[i \dd i + m - 1]$.
  Opposingly, we say that the word $x$ is an \textit{absent word} of
$y$ if it does not occur in $y$.
  The absent word $x$, $m \geq 2$, of $y$ is \textit{minimal} if and only if all its proper factors 
occur in $y$.

  We denote by \SA{} the {\em suffix array} of $y$, that is the array of length $n$
of the starting positions of all sorted suffixes of $y$, i.e.~for all 
$1 \leq  r < n-1$, we have $y[\SA{}[r-1] \dd n-1] < y[\SA{}[r] \dd n - 1]$~\cite{SA}.
  Let \lcp{}$(r, s)$ denote the length of the longest common prefix of
the words $y[\SA{}[r] \dd n - 1]$ and $y[\SA{}[s] \dd n - 1]$, 
for all $0 \leq r,s < n-1$, and $0$ otherwise.
  We denote by \LCP{} the {\em longest common prefix} array of $y$ defined by 
\LCP{}$[r]=\lcp{}(r-1, r)$, for all $1 < r < n-1$, and 
\LCP{}$[0] = 0$. The inverse \iSA{} of the array \SA{} is defined by 
$\iSA{}[\SA{}[r]] = r$, for all $0 \leq r < n-1$.
  \SA{}~\cite{Nong:2009:LSA:1545013.1545570}, \iSA{}, and 
\LCP{}~\cite{indLCP} of $y$ can be computed in time and space $\cO(n)$.

In this article, we consider the following problem.

\defproblem{\MAW}{a word $y$ on $\Sigma$ of length $n$}{all tuples $<a,(i,j)>$, such that word $x$, defined by
$x[0]=a$, $a \in \Sigma$, and $x[1 \dd m-1] = y[i \dd j]$, $m \geq 2$, is a minimal absent word of $y$}
\section{Algorithm \textsf{MAW}}
\label{sec:algo}
  In this section, we present algorithm \textsf{MAW}, an $\cO(n)$-time and $\cO(n)$-space algorithm for finding all minimal absent words in a word of length $n$ using arrays \SA{} and \LCP{}.
  We first explain how we can characterise the minimal absent words; then we introduce how their computation can be done efficiently by using arrays \SA{} and \LCP{}.
  Finally, we present in detail the two main steps of the algorithm.

  A minimal absent word $x[0\dd m-1]$ of a word $y[0\dd n-1]$ is an absent word whose proper factors all occur in $y$.
  Among them, $x_1=x[1\dd m-1]$ and $x_2=x[1\dd m-2]$ occur in $y$; we will focus on these two factors to characterise the minimal absent words.
  To do so, we will consider each occurrence of $x_1$ and $x_2$, and construct the sets of letters that occur just before:
\begin{eqnarray*}
 \textsf{B}(x_1)=\{ y[j-1]: \mbox{ $j$ is the starting position of an occurrence of $x_1$}\}\\
 \textsf{B}(x_2)=\{ y[j-1]: \mbox{ $j$ is the starting position of an occurrence of $x_2$}\}
\end{eqnarray*}

\begin{lemma}
\label{lem:f1}
  Let $x$ and $y$ be two words. Then $x$ is a minimal absent word of $y$ if and only if $x[0]$ is an element of $\textsf{B}(x_2)$ and not of $\textsf{B}(x_1)$, with $x_1=x[1\dd m-1]$ and $x_2=x[1\dd m-2]$.
\end{lemma}
\begin{proof}
    ($\Rightarrow$) Let $x_1$ be a factor of $y$, $x_2$ be the longest proper prefix of $x_1$, and $\textsf{B}(x_1)$ and $\textsf{B}(x_2)$ the sets defined above.
          Further let $p$ be a letter that is in $\textsf{B}(x_2)$ but not in $\textsf{B}(x_1)$. Then, there exists a starting position $j$ of an occurrence of $x_2$
      such that $y[j-1]=p$, so the word $px_2$ occurs at position $j-1$ in $y$. 
          $p$ is not in $\textsf{B}(x_1)$ so $px_1$ does not occur in $x$ and is therefore an absent word of $y$.
          $x_1$ and $px_2$ are factors of $y$, so all the proper factors of $px_1$ occur in $y$, thus $px_1$ is a minimal absent word of $y$.
          
    ($\Leftarrow$) Let $x[0\dd m-1]$ be a minimal absent word of $y$. Its longest proper prefix $x[0\dd m-2]=x[0]x_2$ occurs in $y$, so
      $x[0]$ is in $\textsf{B}(x_2)$. Its longest proper suffix, $x_1$ occurs as well in $y$, but $x=x[0]x_1$ is an absent word of $y$ so it does not occur in $y$ and 
      $x[0]$ is not in $\textsf{B}(x_1)$.\qed
\end{proof}

\begin{lemma} 
\label{lem:f2}
  Let $x$ be a minimal absent word of length $m$ of word $y$ of length $n$. Then there exists an integer $i \in [0:n-1]$ 
such that \textnormal{$y[\SA{}[i]\dd \SA{}[i]+\LCP{}[i]]=x_1$} or \textnormal{$y[\SA{}[i]\dd \SA{}[i]+\LCP{}[i+1]]=x_1$}, where $x_1=x[1\dd m-1]$.
\end{lemma}
\begin{proof}
  Let $j$ be the starting position of an occurrence of $x[0\dd m-2]$ in $y$ and $k$ the starting position of an occurrence of $x_1$ in $y$.
  The suffixes $y[j+1\dd n-1]$ and $y[k\dd n-1]$ share $x_2=x[1\dd m-2]$ as a common prefix.
  As $x$ is an absent word of $y$, this common prefix cannot be extended so $x_2$ is the longest common prefix of those suffixes.
  By using \iSA{}, the inverse suffix array, we have $\lcp{}(\iSA{}[j+1],\iSA{}[k])=m-2$. Let us also note $s_k=\iSA{}[k]$
  and $s_{j+1}=\iSA{}[j+1]$. We then have two possibilities:
\begin{itemize}
  \item if $s_k > s_{j+1}$: \quad for all $s$ in $[s_{j+1}+1:s_k]$, we have
    $\LCP{}[s]\geq m-2$, with equality holding for at least one position.
        Let us define $i=\max\{s \in [s_{j+1}:s_k]:$ $\LCP{}[s]=m-2$ \},
    the maximality of $i$ implies that $i=s_k$ or $\lcp{}(i,s_k)>m-2$ and thus, in both cases
    $y[\SA{}[i]\dd \SA{}[i]+\LCP{}[i]]=x_1$.
  \item if $s_{j+1} > s_k$ : \quad for all $s$ in $[s_k+1:s_{j+1}]$, we have
    $\LCP{}[s]\geq m-2$, with equality holding for at least one position.
        Let us define $i=\min\{s \in [s_k:s_{j+1}]:$ $\LCP{}[s+1]=m-2$ \}, 
        the minimality of $i$ implies $i=s_k$ or $\lcp{}(s_k,i)>m-2$ and thus, in both cases
    $y[\SA{}[i]\dd \SA{}[i]+\LCP{}[i+1]]=x_1$.
\end{itemize}
For an illustration inspect Fig.~\ref{fig:f2}. \qed
\end{proof}
\begin{figure}
\begin{minipage}[c]{0.45\linewidth}
\begin{tikzpicture}
	\draw[->] (0,0) -- (xyz cs:x=3.5);
	\draw[->] (-0.1,3) -- (-0.1,-0.2);

	\foreach \x/\y in {0.05/2,0.25/1.5,0.45/1,0.65/2,0.85/3,1.05/3.5,1.45/2,
	  1.85/2.5,2.05/2,2.25/1,2.65/.5} {
	  \fill[blue!30] (0.05,\x) rectangle +(\y,0.15);
	}
	\foreach \x/\y in {1.25/3,1.65/1,2.45/1.5} {
	  \fill[blue] (0.05,\x) rectangle +(\y,0.15);
	}

	\path[draw=red, line width=1pt] (1.05,0) -- (1.05,2.8);
				
	\node (sk) at (-0.9,1.325) {$s_k$};
	\draw[->] (sk) -- (-0.15,1.325);
	
	\node (si) at (-0.9,1.725) {$i$};
	\draw[->] (si) -- (-0.15,1.725);
	
	\node (sj) at (-0.9,2.525) {$s_{j+1}$};
	\draw[->] (sj) -- (-0.15,2.525);
	
	\node[red] (sk) at (1.05,-0.15) {$m-2$};
	\node (txt) at (1.9,3.7) {{\footnotesize $y[\SA{}[i]\dd \SA{}[i]+\LCP{}[i]]=x_1$}};
	
	\node (txt) at (-0.4,3.2) {\iSA{}};
	\node (txt) at (3.9,0) {\LCP{}};
\end{tikzpicture}

\end{minipage}
\begin{minipage}[c]{0.45\linewidth}
\begin{tikzpicture}
	\draw[->] (0,0) -- (xyz cs:x=3.5);
	\draw[->] (-0.1,3) -- (-0.1,-0.2);
	
	\foreach \x/\y in {0.25/1.5,0.45/1,0.85/3,1.05/3.5,1.45/2,
	  1.65/1,1.85/2.5,2.05/2,2.25/1,2.45/1.5,2.65/.5} {
	  \fill[blue!30] (0.05,\x) rectangle +(\y,0.15);
	}
	\foreach \x/\y in {0.05/2,0.65/2,1.25/3} {
	  \fill[blue] (0.05,\x) rectangle +(\y,0.15);
	}

	\path[draw=red, line width=1pt] (1.05,0) -- (1.05,3);

	\node (sk) at (-0.8,1.325) {$s_k$};
	\draw[->] (sk) -- (-0.15,1.325);
	
	\node (si) at (-0.8,0.725) {$i$};
	\draw[->] (si) -- (-0.15,0.725);
	
	\node (sj) at (-0.8,0.125) {$s_{j+1}$};
	\draw[->] (sj) -- (-0.15,0.125);
	\node[red] at (1.05,-0.15) {$m-2$};
	\node (txt) at (1.9,3.7) {{\footnotesize $y[\SA{}[i]\dd \SA{}[i]+\LCP{}[i+1]]=x_1$}};
	
	\node (txt) at (-0.4,3.2) {\iSA{}};
	
	\node (txt) at (3.9,0) {\LCP{}};
	
\end{tikzpicture}

\end{minipage}
\caption{Illustration of Lemma~\ref{lem:f2}}
\label{fig:f2}
\end{figure}
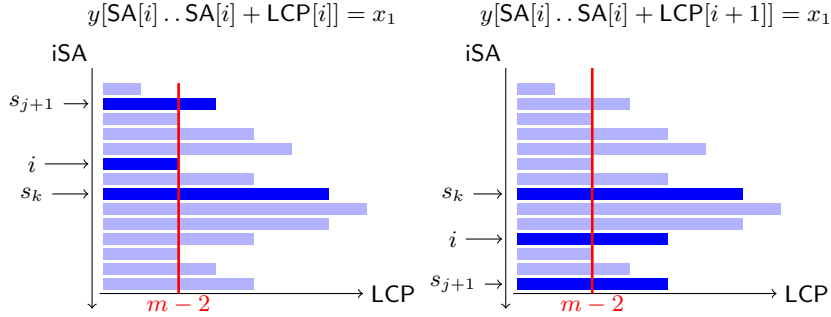
  By Lemma~\ref{lem:f2}, we can compute all minimal absent words of $y$ by examining only the factors $S_{2i}=y[\SA{}[i]\dd \SA{}[i]+\LCP{}[i]]$ and $S_{2i+1}=y[\SA{}[i]\dd \SA{}[i]+\LCP{}[i+1]]$, for all $i$ in $[0:n-1]$. 
  We just need to construct the sets $\textsf{B}_{1}(S_{2i})$, $\textsf{B}_{2}(S_{2i})$ and $\textsf{B}_{1}(S_{2i+1})$, $\textsf{B}_{2}(S_{2i+1})$, where $\textsf{B}_{1}(S_{j})$ (resp.~$\textsf{B}_{2}(S_{j})$) is the set of letters that immediately precede an occurrence of the factor $S_j$ (resp.~the longest proper prefix of $S_j$), for all $j$ in $[0:2n-1]$.
  Then, by Lemma~\ref{lem:f1}, the symmetric difference between $\textsf{B}_{S_{j},1}$ and $\textsf{B}_{S_{j},2}$, for all $j$ in $[0:2n-1]$, gives us all the minimal absent words of $y$.

  Thus the important computational step is to compute these sets of letters efficiently. 
  To do so, we visit twice arrays \SA{} and \LCP{} using another array denoted by \Before{} (resp.~\BeforeLCP{}) to store set $\textsf{B}_{1}(S_j)$ (resp.~$\textsf{B}_{2}(S_j)$), for all $j$ in $[0:2n-1]$. 
  Both arrays \Before{} and \BeforeLCP{} consist of $2n$ elements, where each element is a bit vector of length $\sigma$, the size of the alphabet, corresponding to one bit per alphabet letter.
  While iterating over arrays \SA{} and \LCP{}, we maintain another array denoted by \Interval{}, such that, at the end of each iteration $i$, the $\ell^{th}$ element of \Interval{} stores the set of letters we have encountered before the prefix of length $\ell$ of $y[\SA{}[i]\dd n-1]$. 
  Array \Interval{} consists of ${\displaystyle\max_{ i \in [0:n-1] }}\LCP{}[i]+1$ elements, where each element is a bit vector of length $\sigma$.

  During the first pass, we visit arrays \SA{} and \LCP{} from top to bottom.
  For each $i$ $\in$ $[0:n-1]$, we store in positions $2i$ and $2i+1$ of \Before{} (resp.~\BeforeLCP{}) the set of letters that immediately precede occurrences of $S_{2i}$ and $S_{2i+1}$ (resp.~their longest proper prefixes) whose starting positions appear before position $i$ in \SA{}.
  During the second pass, we go bottom up to complete the sets, which are already stored, with the letters preceding the occurrences whose starting positions appear after position $i$ in \SA{}.
  In order to be efficient, we will maintain a stack structure, denoted by \Lifolcp{}, to store the \LCP{} values of the factors that are prefixes of the one we are currently visiting.
  \begin{myfunction}[H]{11 cm}
  \Algo{\textsf{Top-Down-Pass} (\textnormal{$y$, $n$, \SA{}, \LCP{},  \Before{}, \BeforeLCP{}, $\sigma$})}{
	\Interval{}[$0 \dd \smash{\displaystyle\max_{ i \in [0:n-1]}}$\LCP{}[$i$]][$0 \dd \sigma-1$] $\leftarrow 0$\;
	\Lifolcp{}.\textsf{push}(0)\;
	\ForEach{\textnormal{$i \in [0:n-1]$}}{
		\If{\textnormal{$i>0$ \textbf{and} $\LCP{}[i] < \LCP{}[i-1]$}}
		{
			\While{\textnormal{$\Lifolcp{}.\textsf{top}()>\LCP{}[i]$}}
			{
				$\textsf{proxa} \leftarrow \Lifolcp{}.\textsf{pop}()$\;
				$\Interval[\textsf{proxa}][0 \dd \sigma - 1]\leftarrow 0$\;
			}
			\If{\textnormal{$\Lifolcp{}.\textsf{top}() < \LCP{}[i]$}}{$\Interval{}[\LCP{}[i]]\leftarrow \Interval{}[\textsf{proxa}]$\;}
			$\Before{}[2i-1]\leftarrow \Interval{}[\textsf{proxa}]$; $\BeforeLCP{}[2i-1]\leftarrow \Interval{}[\LCP{}[i]]$\;
		} 
	\If{\textnormal{ $\SA{}[i]>0$ }}{
		$u \leftarrow y[\SA{}[i]-1]$; $\textsf{value} \leftarrow \Lifolcp{}.\textsf{top}()$\;
		\While{\textnormal{$\Interval{}[\textsf{value}][u]=0$}}
		 {
		      $\Interval{}[\textsf{value}][u] \leftarrow 1$; $\textsf{value} \leftarrow \Lifolcp{}.\textsf{next}()$\;
		 }
		$\Interval{}[\LCP{}[i]][u] \leftarrow 1$\;
		$\Before{}[2i][u] \leftarrow 1$; $\Before{}[2i+1][u] \leftarrow 1$\;
		$\BeforeLCP{}[2i][u] \leftarrow 1$; $\BeforeLCP{}[2i+1][u] \leftarrow 1$\;
	}
	\If{\textnormal{$i>0$ \textbf{and} $\LCP{}[i]>0$ \textbf{ and } $\SA{}[i-1]>0$}}{
		$v \leftarrow y[\SA{}[i-1]-1]$\;
		$\Interval{}[\LCP{}[i]][v] \leftarrow 1$ \;
	}
	$\BeforeLCP{}[2i] \leftarrow \Interval{}[\LCP{}[i]]$\;
	\lIf{\textnormal{$\Lifolcp{}.\textsf{top}() \neq \LCP{}[i]$}}{$\Lifolcp{}.\textsf{push}(\LCP{}[i])$}
	}
	}
  \end{myfunction}
\subsection{Top-down Pass}

  Each iteration of the top-down pass consists of two steps. 
  In the first step, we visit \Lifolcp{} from the top and for each \LCP{} value read we set to zero the corresponding element of \Interval{}; then we remove this value from the stack. 
  We stop when we reach a value smaller or equal to $\LCP{}[i]$.
  We do this as the corresponding factors are not prefixes of $y[\SA{}[i]\dd n-1]$, nor will they be prefixes in the remaining suffixes.
  We push at most one value onto the stack \Lifolcp{} per iteration, so, in total, there are $n$ times we will set an element of \Interval{} to zero.
  This step requires time and space $\cO(n\sigma)$.
  
  For the second step, we update the elements that correspond to factors in the suffix array with an \LCP{} value less than $\LCP{}[i]$.
  To do so, we visit the stack \Lifolcp{} top-down and, for each \LCP{} value read, we add the letter 
  $y[\SA{}[i]-1]$ to the corresponding element of \Interval{} until we reach a value whose element already contains it.
  This ensures that, for each value read, the corresponding element of \Interval{} has no more than $\sigma$ letters added.
  As we consider at most $n$ values, this step requires time and space $\cO(n\sigma)$. For an example, see Table~\ref{tab:td}.
\begin{table}[!t]
\begin{center}
\subfloat[{\footnotesize}]{
\scalebox{0.75}{
\begin{tabular}{|*{3}{c|}}\hline
 $j$ & \Before{} & \BeforeLCP{} \\ \hline
 0 & 00 & 00 \\ \hline
 1 & 00 & 00 \\ \hline
 2 & 10 & 10 \\ \hline
 3 & 10 & 10 \\ \hline
 4 & 01 & 11 \\ \hline
 5 & 11 & 11 \\ \hline
 6 & 01 & 11 \\ \hline
 7 & 11 & 11 \\ \hline
 8 & 01 & 11 \\ \hline
 9 & 01 & 01 \\ \hline
 10 & 10 & 11 \\ \hline
 11 & 10 & 10 \\ \hline
 12 & 10 & 10 \\ \hline
 13 & 10 & 11 \\ \hline
 14 & 10 & 11 \\ \hline
\end{tabular}
}
}
\qquad
\subfloat[{\footnotesize}]{
\scalebox{0.75}{
\begin{tabular}{|*{3}{c|}*{4}{c}*{6}{c|}}\hline
 $i$&\LCP{}&\SA{}&\multicolumn{5}{c|}{Factor}&\Interval{}[0]&\Interval{}[1]&\Interval{}[2]&\Interval{}[3]&\Interval{}[4]\\
\hline 0&0&0&\cellcolor{Thistle}\texttt{A}& & & & &00&00&00&00&00\\ \hline
 & & &\cellcolor{Apricot}\texttt{A}&\cellcolor{Thistle}\texttt{A}& & & & & & & & \\\hline
1&1&1&\cellcolor{Apricot}\texttt{A}&\cellcolor{Thistle}\texttt{B}& & & &10&10&00&00&00\\ \hline
 & & &\cellcolor{Apricot}\texttt{A}&\cellcolor{Apricot}\texttt{B}&\cellcolor{Apricot}\texttt{A}&\cellcolor{Apricot}\texttt{B}&\cellcolor{Thistle}\texttt{A}& & & & & \\ \hline
2&4&3&\cellcolor{Apricot}\texttt{A}&\cellcolor{Apricot}\texttt{B}&\cellcolor{Apricot}\texttt{A}&\cellcolor{Apricot}\texttt{B}&\cellcolor{Thistle}\texttt{B}&11&11&00&00&11\\ \hline
 & & &\cellcolor{Apricot}\texttt{A}&\cellcolor{Apricot}\texttt{B}&\cellcolor{Thistle}\texttt{A}& & & & & & & \\ \hline
3&2&5&\cellcolor{Apricot}\texttt{A}&\cellcolor{Apricot}\texttt{B}&\cellcolor{Thistle}\texttt{B}& & &11&11&11&00&00\\ \hline
 & & &\cellcolor{Thistle}\texttt{A}& & & & & & & & & \\ \hline
4&0&7&\cellcolor{Thistle}\texttt{B}& & & & &11&00&00&00&00\\ \hline
 & & &\cellcolor{Apricot}\texttt{B}&\cellcolor{Thistle}
& & & & & & & &\\ \hline
5&1&2&\cellcolor{Apricot}\texttt{B}&\cellcolor{Thistle}\texttt{A}& & & &11&11&00&00&00\\ \hline
 & & &\cellcolor{Apricot}\texttt{B}&\cellcolor{Apricot}\texttt{A}&\cellcolor{Apricot}\texttt{B}&\cellcolor{Thistle}\texttt{A}& & & & & &\\ \hline
6&3&4&\cellcolor{Apricot}\texttt{B}&\cellcolor{Apricot}\texttt{A}&\cellcolor{Apricot}\texttt{B}&\cellcolor{Thistle}\texttt{B}& &11&11&00&10&00\\ \hline
 & & &\cellcolor{Apricot}\texttt{B}&\cellcolor{Thistle}\texttt{A}& & & & & & & &\\ \hline
7&1&6&\cellcolor{Apricot}\texttt{B}&\cellcolor{Thistle}\texttt{B}& & & &11&11&00&00&00\\ \hline
\end{tabular}
}
}
\end{center}
  \caption{(a) Arrays \Before{} and \BeforeLCP{} obtained after the top-down pass for word $y=\texttt{AABABABB}$; (b) Elements of array \Interval{}
at the end of each iteration of the top-down pass. Factors $S_{j}$ are in orange and violet; their longest proper prefixes are in orange only.}
\label{tab:td}
\end{table}
\subsection{Bottom-up Pass}
  Intuitively, the idea behind the bottom-up pass is the same as in the top-down pass except that in this instance, as we start from the bottom, the suffix $y[\SA{}[i]\dd n-1]$ can share more than its prefix of length $\LCP{}[i]$ with the previous suffixes in \SA{}.
  Therefore we may need the elements of \Interval{} that correspond to factors with an \LCP{} value greater than $\LCP{}[i]$ to correctly compute the arrays \Before{} and \BeforeLCP{}.
  To achieve this, we maintain another stack \Liforem{} to copy the values from \Lifolcp{} that are greater than $\LCP{}[i]$.
  This extra stack allows us to keep in \Lifolcp{} only values that are smaller or equal to $\LCP[i]$ without losing the additional information we require to correctly compute \Before{} and \BeforeLCP{}.
  At the end of the iteration, we will set to zero each element corresponding to a value in \Liforem{} and empty the stack.
  Thus to set an element of \Interval{} to zero requires two operations more than in the first pass.
  As we consider at most $n$ values, this step requires time and space $\cO(n\sigma)$.

  \begin{myfunction}[H]{11cm}
\Algo{\textsf{Bottom-Up-Pass}($n$, \textnormal{\SA{}, \LCP{},  \Before{}, \BeforeLCP{}, $\Sigma$, $\sigma$})}
{ 
	\Interval{}[$0 \dd \smash{\displaystyle\max_{ i \in [0:n-1]}}$\LCP{}[$i$]][$0 \dd \sigma-1$] $\leftarrow 0$\;
	\Lifolcp{}.\textsf{push}(0)\;
	\ForEach{\textnormal{$i \in [n-1:0]$}}
	{
		\textsf{proxa} $\leftarrow \LCP{}[i]+1$; \textsf{proxb} $\leftarrow 1$\;
		\If{ \textnormal{$i<n-1$ \textbf{and} $\LCP{}[i] < \LCP{}[i+1]$}}
		{
			\While{\textnormal{$\Lifolcp{}.\textsf{top}() > \LCP{}[i]$}}
			{
				$\textsf{proxa} \leftarrow \Lifolcp{}.\textsf{pop}()$\;
				$\Liforem{}.\textsf{push}(\textsf{proxa})$\;
			}
			\If{\textnormal{$\Lifolcp{}.\textsf{top}()<\LCP{}[i]$}}{$\Interval{}[\LCP{}[i]]\leftarrow \Interval{}[\textsf{proxa}]$}
		}
		\ForEach{\textnormal{$k \in \Sigma: \Before{}[2i][k]=1$}}
		{
			$\textsf{value} \leftarrow \Lifolcp{}.\textsf{top}()$\;
			\While{\textnormal{$\Interval{}[\textsf{value}][k]=0$}}
			{
				$\Interval{}[\textsf{value}][k] \leftarrow 1$; $\textsf{value} \leftarrow \Lifolcp{}.\textsf{next}()$;
			}
			$\Interval{}[\LCP{}[i]][k] \leftarrow 1$\;
		}
		\begin{tabbing}
		\=$\BeforeLCP{}[2i+1]$\= $\leftarrow \BeforeLCP[2i+1]$ \= $\textbf{or}~\Interval[\LCP[i+1]]$\kill
		\> $\BeforeLCP{}[2i]$\> $\leftarrow \BeforeLCP{}[2i]$\> $\textbf{bit-or}~\Interval{}[\LCP{}[i]]$;\\
		\> $\BeforeLCP{}[2i+1]$\> $\leftarrow \BeforeLCP{}[2i+1]$\> $\textbf{bit-or}~\Interval{}[\LCP{}[i+1]]$;\\
		\> $\Before{}[2i+1]$\> $\leftarrow \Before{}[2i+1]$\> $\textbf{bit-or}~\Interval{}[\textsf{proxb}]$;
		\end{tabbing}
		\textsf{proxb} $\leftarrow$ \textsf{proxa}\;
		$\Before{}[2i] \leftarrow \Before{}[2i] \textbf{ bit-or } \Interval{}[\textsf{proxa}]$\;
		\While{\textnormal{\Liforem{} not empty}}
		{
			$\textsf{value}\leftarrow \Liforem{}.\textsf{pop}()$; $\Interval{}[\textsf{value}][0 \dd \sigma - 1]\leftarrow 0$;
		}
	\lIf{\textnormal{$\Lifolcp{}.\textsf{top}() \neq \LCP{}[i]$}}{$\Lifolcp{}.\textsf{push}(\LCP{}[i])$}
	}
}
\end{myfunction}

  Another difference between the top-down and bottom-up passes is that in order to retain the information computed in the first pass, the second step is performed for each letter in $\Before{}[2i]$. 
  As, for each \LCP{} value read, we still add a letter only if is not already contained in the corresponding element of \Interval{}, no more than $\sigma$ letters are added.
  Thus this step requires time and space $\cO(n\sigma)$. For an example, see Table~\ref{tab:bu}.
  
  \begin{table}
\begin{center}
\subfloat[{\footnotesize}]{
\scalebox{0.75}{
\begin{tabular}{|*{3}{>{\small}c|}}\hline
 $j$ & \Before{} & \BeforeLCP{} \\ \hline
 0 & 11 & 11 \\ \hline
 1 & 00 & 11 \\ \hline
 2 & 11 & 11 \\ \hline
 3 & 10 & 11 \\ \hline
 4 & 01 & 11 \\ \hline
 5 & 11 & 11 \\ \hline
 6 & 01 & 11 \\ \hline
 7 & 11 & 11 \\ \hline
 8 & 11 & 11 \\ \hline
 9 & 01 & 11 \\ \hline
 10 & 10 & 11 \\ \hline
 11 & 10 & 10 \\ \hline
 12 & 10 & 10 \\ \hline
 13 & 10 & 11 \\ \hline
 14 & 10 & 11 \\ \hline
\end{tabular}
}
}
\qquad
\subfloat[{\footnotesize}]{
\scalebox{0.75}{
\begin{tabular}{|*{3}{c|}*{4}{c}*{6}{c|}}\hline
 $i$&\LCP{}&\SA{}&\multicolumn{5}{c|}{Factor}&\Interval{}[0]&\Interval{}[1]&\Interval{}[2]&\Interval{}[3]&\Interval{}[4]\\ \hline 
7&1&6&\cellcolor{Apricot}\texttt{B}&\cellcolor{Thistle}\texttt{B}& & & &10&10&00&00&00\\ \hline
 & & &\cellcolor{Apricot}\texttt{B}&\cellcolor{Thistle}\texttt{A}& & & & & & & & \\\hline
6&3&4&\cellcolor{Apricot}\texttt{B}&\cellcolor{Apricot}\texttt{A}&\cellcolor{Apricot}\texttt{B}&\cellcolor{Thistle}\texttt{B}& &10&10&00&10&00\\ \hline
 & & &\cellcolor{Apricot}\texttt{B}&\cellcolor{Apricot}\texttt{A}&\cellcolor{Apricot}\texttt{B}&\cellcolor{Thistle}\texttt{A}& & & & & & \\ \hline
5&1&2&\cellcolor{Apricot}\texttt{B}&\cellcolor{Thistle}\texttt{A}& & & &10&10&00&00&00\\ \hline
 & & &\cellcolor{Apricot}\texttt{B}&\cellcolor{Thistle}
& & & & & & & & \\ \hline
4&0&7&\cellcolor{Thistle}\texttt{B}& & & & &11&00&00&00&00\\ \hline
 & & &\cellcolor{Thistle}\texttt{A}& & & & & & & & & \\ \hline
3&2&5&\cellcolor{Apricot}\texttt{A}&\cellcolor{Apricot}\texttt{B}&\cellcolor{Thistle}\texttt{B}& & &11&00&01&00&00\\ \hline
 & & &\cellcolor{Apricot}\texttt{A}&\cellcolor{Apricot}\texttt{B}&\cellcolor{Thistle}\texttt{A}& & & & & & & \\ \hline
2&4&3&\cellcolor{Apricot}\texttt{A}&\cellcolor{Apricot}\texttt{B}&\cellcolor{Apricot}\texttt{A}&\cellcolor{Apricot}\texttt{B}&\cellcolor{Thistle}\texttt{B}&11&00&01&00&01\\ \hline
 & & &\cellcolor{Apricot}\texttt{A}&\cellcolor{Apricot}\texttt{B}&\cellcolor{Apricot}\texttt{A}&\cellcolor{Apricot}\texttt{B}&\cellcolor{Thistle}\texttt{A}& & & & &\\ \hline
1&1&1&\cellcolor{Apricot}\texttt{A}&\cellcolor{Thistle}\texttt{B}& & & &11&11&00&00&00\\ \hline
 & & &\cellcolor{Apricot}\texttt{A}&\cellcolor{Thistle}\texttt{A}& & & & & & & &\\ \hline
0&0&0&\cellcolor{Thistle}\texttt{A}& & & & &11&00&00&00&00\\ \hline
\end{tabular}
}
}
  \end{center}
  \caption{(a) Arrays \Before{} and \BeforeLCP{} obtained after the bottom-up pass for word $y=\texttt{AABABABB}$; (b) Elements of array \Interval{} at the end of each iteration of the bottom-up pass. Factors $S_{j}$ are in orange and violet; their longest proper prefixes are in orange only.}
\label{tab:bu}
\end{table}

  Once we have computed arrays \Before{} and \BeforeLCP{}, we need to compare each element.
  If there is a symmetric difference, by Lemma~\ref{lem:f1}, we can construct a minimal absent word. For an example, see Table~\ref{tab:maws}. 
  To ensure that we can report the minimal absent words in linear time, we must be able to report each one in constant time. 
  To achieve this, we can represent them as a tuple $<a,(i,j)>$, where for some word $x$ of length $m \geq 2$ that is a minimal absent word of $y$, 
  the following holds: $x[0]=a$ and $x[1 \dd m-1] = y[i \dd j]$. Lemma~\ref{lem:f2} ensures us to be exhaustive. Therefore we obtain the following result.
\begin{theorem}
Algorithm \textsf{MAW} solves problem $\MAW$ in time and space $\cO(n)$.
\label{the:main}
\end{theorem}
\begin{table}
\begin{center}
\scalebox{0.75}{
\begin{tabular}{|*{3}{c|}*{5}{c|}l|l|}\hline
 $j$ & \Before{} & \BeforeLCP{} &\multicolumn{5}{c|}{Factor}& Minimal absent words &Tuple representation\\ \hline
 0 & 11 & 11 &\cellcolor{Thistle}\texttt{A}& & & & & &\\ \hline
 1 & 00 & 11 &\cellcolor{Apricot}\texttt{A}&\cellcolor{Thistle}\texttt{A}& & & & \texttt{AAA, BAA} &$<$\texttt{A}$,(0,1)>$, $<$\texttt{B}$,(0,1)>$\\ \hline
 2 & 11 & 11 &\cellcolor{Apricot}\texttt{A}&\cellcolor{Thistle}\texttt{B}& & & & &\\ \hline
 3 & 10 & 11 &\cellcolor{Apricot}\texttt{A}&\cellcolor{Apricot}\texttt{B}&\cellcolor{Apricot}\texttt{A}&\cellcolor{Apricot}\texttt{B}&\cellcolor{Thistle}\texttt{A}&\texttt{BABABA}&$<$\texttt{B}$,(1,5)>$ \\ \hline
 4 & 01 & 11 &\cellcolor{Apricot}\texttt{A}&\cellcolor{Apricot}\texttt{B}&\cellcolor{Apricot}\texttt{A}&\cellcolor{Apricot}\texttt{B}&\cellcolor{Thistle}\texttt{B}&\texttt{AABABB}&$<$\texttt{A}$,(3,7)>$ \\ \hline
 5 & 11 & 11&\cellcolor{Apricot}\texttt{A}&\cellcolor{Apricot}\texttt{B}&\cellcolor{Thistle}\texttt{A}& & & & \\ \hline
 6 & 01 & 11 &\cellcolor{Apricot}\texttt{A}&\cellcolor{Apricot}\texttt{B}&\cellcolor{Thistle}\texttt{B}& & &\texttt{AABB} &$<$\texttt{A}$,(5,7)>$\\ \hline
 7 & 11 & 11 &\cellcolor{Thistle}\texttt{A}& & & & & & \\ \hline
 8 & 11 & 11 &\cellcolor{Thistle}\texttt{B}& & & & & &\\ \hline
 9 & 01 & 11 &\cellcolor{Apricot}\texttt{B}&\cellcolor{Thistle}
&\multicolumn{5}{l|}{We do not consider this row as it corresponds to the end of the word $y$} \\ \hline
 10 & 10 & 11 &\cellcolor{Apricot}\texttt{B}&\cellcolor{Thistle}\texttt{A}& & & &\texttt{BBA} &$<$\texttt{B}$,(2,3)>$\\ \hline
 11 & 10 & 10 &\cellcolor{Apricot}\texttt{B}&\cellcolor{Apricot}\texttt{A}&\cellcolor{Apricot}\texttt{B}&\cellcolor{Thistle}\texttt{A}& & &\\ \hline
 12 & 10 & 10 &\cellcolor{Apricot}\texttt{B}&\cellcolor{Apricot}\texttt{A}&\cellcolor{Apricot}\texttt{B}&\cellcolor{Thistle}\texttt{B}& & &\\ \hline
 13 & 10 & 11 & \cellcolor{Apricot}\texttt{B}&\cellcolor{Thistle}\texttt{A}& & & &\texttt{BBA}&{This is a duplicate so we ignore it}\\ \hline
 14 & 10 & 11 &\cellcolor{Apricot}\texttt{B}&\cellcolor{Thistle}\texttt{B}& & & &\texttt{BBB}& $<$\texttt{B}$,(6,7)>$\\ \hline
\end{tabular}
}
  \end{center}
  \caption{Minimal absent words of word $y=\texttt{AABABABB}$; we find seven minimal absent words $\{\texttt{AAA, AABABB, AABB, BAA, BABABA, BBA, BBB}\}$}
  \label{tab:maws}
\end{table}
\section{Experimental Results}
\label{sec:exp}
  We implemented algorithm $\textsf{MAW}$ as a programme to compute all minimal absent words of a given sequence. 
  The programme was implemented in the $\textsf{C}$ programming language and developed under GNU/Linux operating system. 
  It takes as input arguments a file in (Multi)FASTA format and the minimal and maximal length of minimal absent words to be outputted; and then 
  produces a file with all minimal absent words of length within this range as output.
  The implementation is distributed under the GNU General Public License (GPL), and it is available at \url{http://github.com/solonas13/maw}, which is set up for maintaining the source code and the man-page documentation.
  The experiments were conducted on a Desktop PC using one core of Intel Xeon E5540 CPU at 2.5 GHz and 32GB of main memory under 64-bit GNU/Linux.
  We considered the genomes of thirteen bacteria and four case-study eukaryotes (Table~\ref{tab:data}), all obtained from the NCBI database (\url{ftp://ftp.ncbi.nih.gov/genomes/}).

\begin{table}[!ht]
\begin{center}
\scalebox{0.85}{
\begin{tabular}{lll} \hline
Species							& Abbreviation    	& Genome reference\\ \hline
{\bf Bacteria}     					&   	      		& 	\\ 
{\em Bacillus anthracis strain Ames}     		& Ba  	      		& NC003997	\\
{\em Bacillus subtilis strain 168}    			& Bs  	      		& NC000964	\\
{\em Escherichia coli strain K-12 substrain MG1655} 	& Ec          		& NC000913     	\\ 
{\em Haemophilus influenzae strain Rd KW20}     	& Hi  	      		& NC000907	\\
{\em Helicobacter pylori strain 26695}     		& Hp  	      		& NC000915 	\\
{\em Lactobacillus casei strain BL23}    		& Lc          		& NC010999      \\ 
{\em Lactococcus lactis strain Il1403}    		& Ll          		& NC002662      \\ 
{\em Mycoplasma genitalium strain G37}     		& Mg  	      		& NC000908	\\
{\em Staphylococcus aureus strain N315}   		& Sa  	      		& NC002745	\\ 
{\em Streptococcus pneumoniae strain CGSP14} 		& Sp          		& NC010582     	\\ 
{\em Xanthomonas campestris strain 8004}     		& Xc  	      		& NC007086	\\ \hline
{\bf Eukaryotes}     					&   	      		& 		\\ 
{\em Arabidopsis thaliana (thale cress)}		& At			& AGI release 7.2 \\
{\em Drosophila melanogaster (fruit fly})		& Dm			& FlyBase release 5 \\
{\em Homo sapiens (human)}    				& Hs          		& build 38      \\ 
{\em Mus musculus (mouse)}    	 			& Mm  	      		& build 38 	\\\hline

\end{tabular}
}
\end{center}
\caption{Species selected for this work with reference to the respective
abbreviation and identification of genome sequence data by accession 
number for bacteria or genome assembly project for eukaryotes}
\label{tab:data}
\end{table}

  To test the correctness of our implementation, we compared it against
the implementation of Pinho et al.~\cite{Pinho2009}, which we denote here
by \textsf{PFG}.
  In particular, we counted the number of minimal absent words, for 
lengths $11$, $14$, $17$, and $24$, in the genomes of the thirteen 
bacteria listed in Table~\ref{tab:data}.
  We considered only the $5'\rightarrow3'$ DNA strand. 
  Table~\ref{tab:correct} depicts the number of minimal absent words in 
these sequences. 
  We denote by $\textsf{M}_{11}$, $\textsf{M}_{14}$, $\textsf{M}_{17}$, 
and $\textsf{M}_{24}$ the size of the resulting sets of minimal absent
words for lengths $11$, $14$, $17$, and $24$ respectively. 
  Identical number of minimal absent words for these lengths were also
reported by \textsf{PFG}, suggesting that our implementation is correct.

\begin{table}[!ht]
\vspace{0.25cm}
\begin{center}
\scalebox{0.85}{
\begin{tabular}{llllll} \hline
Species& Genome size (bp)    		&$\textsf{M}_{11}$	& $\textsf{M}_{14}$	& $\textsf{M}_{17}$	& $\textsf{M}_{24}$\\ \hline
Ba     & 5,227,293 	      		& 1,113,398	&  1,001,357		& 32,432		&  46	         \\
Bs     & 4,214,630  	      		& 951,273	&  1,703,309		& 86,372		&  226		 \\ 
Ec     & 4,639,675          		& 1,072,074     &  1,125,653		& 36,395		&  247		 \\ 
Hi     & 1,830,023  	      		& 722,860	&  294,353		& 12,158		&  91		 \\
Hp     & 1,667,825  	      		& 564,308 	&  336,122		& 19,276		&  75		 \\
Lc     & 3,079,196         		& 1,126,363     &  502,861		& 13,083		&  246		 \\
Ll     & 2,365,589  	      		& 764,006 	&  507,490		& 25,667		&  183		 \\
Mg     & 1,664,957          		& 246,342      	&  66,324		& 2,737			&  28		 \\
Sa     & 2,814,816         		& 755,483      	&  704,147		& 32,054		&  138		 \\
Sp     & 2,209,198  	      		& 904,815 	&  327,713		& 10,390		&  234		 \\
Xc     & 5,148,708          		& 804,034      	&  1,746,214		& 179,346		&  633		 \\\hline
\end{tabular}
}
\end{center}
\caption{Number of minimal absent words of lengths $11$, $14$, $17$, 
and $24$ in the genomes of thirteen bacteria.}
\label{tab:correct}
\end{table}

  To evaluate the efficiency of our implementation, we compared it 
against the corresponding performance of \textsf{PFG}, which is currently
the fastest available implementation for computing minimal absent words.
  We computed all minimal absent words for each chromosome sequence of 
the genomes of the four eukaryotes listed in Table~\ref{tab:data}.
  We considered both the $5'\rightarrow3'$ and the $3'\rightarrow5'$ 
DNA strands.
  Tables~\ref{tab:time1} and~\ref{tab:time2} depict elapsed-time 
comparisons of \textsf{MAW} and \textsf{PFG}.
  \textsf{MAW} scales linearly and is the fastest in {\em all} cases.
  It accelerates the computations by more than a factor of $2$, when 
the length of the sequences grows, compared to \textsf{PFG}.
  \textsf{MAW} also reduces the memory requirements by a factor of $5$ compared to \textsf{PFG}.
  The maximum allocated memory (per task) was 6GB for \textsf{MAW} and 30GB for \textsf{PFG}.

\begin{table}[!ht]
\vspace{0.25cm}
\begin{center}
\subfloat[At]
{
\scalebox{0.85}{
\begin{tabular}{llll} \hline
Chromosome			& Size (bp)			& \textsf{MAW} (s)	& \textsf{PFG} (s)\\ \hline
1	     			&30,427,671				& 40.20	      		& 51.90	\\
2     				&19,698,289				& 25.86	      		& 32.94	\\
3     				&23,459,830				& 30.84	      		& 42.30	\\
4     				&18,585,056				& 24.65	      		& 31.42	\\
5     				&26,975,502				& 35.38	      		& 48.91	\\ \hline
\end{tabular}
}
}
\qquad
\subfloat[Dm]
{
\scalebox{0.85}{
\begin{tabular}{llll} \hline
Chromosome			& Size (bp)		& \textsf{MAW} (s)	& \textsf{PFG} (s)\\ \hline 
2L     				&23,011,544			& 30.01	      		& 40.85	\\
2R     				&21,146,708			& 27.52	      		& 38.38	\\
3L     				&24,543,557			& 32.00	      		& 45.13	\\
3R     				&27,905,053			& 36.44	      		& 48.36	\\
X     				&22,422,827			& 29.38	      		& 40.09	\\ \hline
\end{tabular}
}
}
\end{center}
\caption{Elapsed-time comparison of \textsf{MAW} and \textsf{PFG} 
for computing all minimal absent words in the genome of 
{\em Arabidopsis thaliana} and {\em Drosophila melanogaster}}
\label{tab:time1}
\end{table}

\begin{table}[!ht]
\vspace{0.25cm}
\begin{center}
\subfloat[Hs]
{
\scalebox{0.85}{
\begin{tabular}{llll} \hline
Chromosome			& Size (bp)	& \textsf{MAW} (s)	& \textsf{PFG} (s)\\ \hline
1     				&248,956,422	& 426.39  	      		& 972.52	\\
2     				&242,193,529	& 423.19  	      		& 772.89	\\
3     				&198,295,559	& 353.60  	      		& 645.45	\\
4     				&190,214,555	& 339.02  	      		& 616.26	\\
5     				&181,538,259	& 342.53  	      		& 577.05	\\   
6     				&170,805,979	& 299.72  	      		& 538.34	\\
7     				&159,345,973	& 305.26  	      		& 491.32	\\
8     				&145,138,636	& 254.17  	      		& 437.18	\\
9     				&138,394,717	& 235.14  	      		& 356.08	\\
10     				&133,797,422	& 235.38  	      		& 392.45	\\
11     				&135,086,622	& 236.80  	      		& 379.15	\\
12     				&133,275,309	& 235.14  	      		& 390.46	\\
13     				&114,364,328	& 191.64  	      		& 269.52	\\
14     				&107,043,718	& 178.00  	      		& 240.93	\\   
15     				&101,991,189	& 167.89  	      		& 222.98	\\
16     				&90,338,345	& 153.07  	      		& 198.49	\\
17     				&83,257,441	& 144.32  	      		& 207.02	\\
18     				&80,373,285	& 137.68  	      		& 199.44	\\
19     				&58,617,616	& 100.95  	      		& 126.82	\\
20     				&64,444,167	& 109.80  	      		& 144.83	\\
21     				&46,709,983	& 74.60  	      		& 74.65		\\
22     				&50,818,468	& 70.49  	      		& 73.34		\\ 
X     				&156,040,895	& 275.14  	      		& 457.2		\\ 
Y     				&57,227,415	& 60.85  	      		& 62.34		\\ \hline
\end{tabular}
}
}
\qquad
\subfloat[Mm]
{
\scalebox{0.85}{
\begin{tabular}{llll} \hline
Chromosome			& Size (bp)	& \textsf{MAW} (s)	& \textsf{PFG} (s)\\ \hline
1     				&197,195,432	& 340.59  	      		& 599.86	\\
2     				&181,748,087	& 316.17  	      		& 578.2	\\
3     				&159,599,783	& 274.46  	      		& 506.73	\\
4     				&155,630,120	& 266.67  	      		& 473.97	\\
5     				&152,537,259	& 260.50  	      		& 424.24	\\   
6     				&149,517,037	& 256.36  	      		& 455.11	\\
7     				&152,524,553	& 257.65  	      		& 413.37	\\
8     				&131,738,871	& 223.09  	      		& 344.92	\\
9     				&124,076,172	& 210.37  	      		& 334.25	\\
10     				&129,993,255	& 222.36  	      		& 363.34	\\
11     				&121,843,856	& 208.55  	      		& 324.54	\\
12     				&121,257,530	& 205.09  	      		& 324.79	\\
13     				&120,284,312	& 204.80  	      		& 314.56	\\
14     				&125,194,864	& 212.59  	      		& 336.49	\\   
15     				&103,494,974	& 175.21  	      		& 265.92	\\
16     				&98,319,150	& 166.10  	      		& 249.03	\\
17     				&95,272,651	& 160.70  	      		& 232.79	\\
18     				&90,772,031	& 153.40  	      		& 223.56	\\
19     				&61,342,430	& 101.89  	      		& 125.85	\\
X     				&166,650,296	& 282.21  	      		& 503.98	\\ 
Y     				&91,744,698	& 141.79  	      		& 251	\\ \hline
\end{tabular}
}
}
\end{center}
\caption{Elapsed-time comparison of \textsf{MAW} and \textsf{PFG} 
for computing all minimal absent words in the genome of {\em Homo
Sapiens} and {\em Mus musculus}}
\label{tab:time2}
\end{table}

  To further evaluate the efficiency of our implementation, we compared
it against the corresponding performance of \textsf{PFG} using synthetic
data.
  As basic dataset we used chromosome 1 of Hs. 
  We created five instances $\textsf{S}_{1}$, $\textsf{S}_{2}$, 
$\textsf{S}_{3}$, $\textsf{S}_{4}$, and $\textsf{S}_{5}$ of this
sequence by randomly choosing 10\%, 20\%, 30\%, 40\%, and 50\% of 
the positions, respectively, and randomly replacing the corresponding
letters to one of the four letters of the DNA alphabet.
  We computed all minimal absent words for each instance. 
  We considered both the $5'\rightarrow3'$ and the $3'\rightarrow5'$ 
DNA strands.
  Table~\ref{tab:time3} depicts elapsed-time comparisons of 
\textsf{MAW} and \textsf{PFG}.
  \textsf{MAW} is the fastest in {\em all} cases.

\begin{table}[!ht]
\vspace{0.25cm}
\begin{center}
\scalebox{0.85}{
\begin{tabular}{llll} \hline
Sequence					& Size (bp)			& \textsf{MAW} (s)	& \textsf{PFG} (s)\\ \hline
$\textsf{S}_{1}$	     			&248,956,422			& 435.63      		& 746.93\\
$\textsf{S}_{2}$     				&248,956,422			& 438.52      		& 733.69\\
$\textsf{S}_{3}$     				&248,956,422			& 444.62      		& 726.34\\
$\textsf{S}_{4}$     				&248,956,422			& 444.06      		& 743.29\\
$\textsf{S}_{5}$     				&248,956,422			& 449.25      		& 741.01\\ \hline
\end{tabular}
}
\end{center}
\caption{Elapsed-time comparison of \textsf{MAW} and \textsf{PFG}
for computing all minimal absent words in synthetic data}
\label{tab:time3}
\end{table}
\section{Final Remarks}
\label{sec:conc}
  We presented the first $\cO(n)$-time and $\cO(n)$-space algorithm for computing all minimal absent words based on the construction of suffix arrays. 
  Experimental results show that the respective implementation outperforms existing tools.
\bibliographystyle{plain}
\bibliography{references}
\end{document}